\documentclass[conference]{IEEEtran}
\IEEEoverridecommandlockouts

\usepackage{cite}

\usepackage{amsfonts,amssymb,amsthm,mathrsfs,bm,bbm}
\usepackage{graphicx}
\usepackage{overpic}
\usepackage{textcomp}
\usepackage{xcolor}
\usepackage{enumerate}
\usepackage{algpseudocode}
\usepackage{amsmath}

\ifCLASSOPTIONcompsoc 
\usepackage[caption=false,font=normalsize,labelfon t=sf,textfont=sf]{subfig} 
\else 
\usepackage[caption=false,font=footnotesize]{subfig} \fi 

\newcommand{\eqdef}{:=}

\renewcommand{\vec}[1]{\bm{#1}}		
\newcommand{\real}{\stdset{R}} 		
\newcommand{\rvec}[1]{\mathbbm{#1}} 		
\newcommand{\E}{\mathsf{E}}		
\newcommand{\Var}{\mathsf{V}}			
\newcommand{\stdset}[1]{\mathbbmss{#1}}	
\newcommand{\set}[1]{\mathcal{#1}}		
\newcommand{\refeq}[1]{(\ref{#1})}		%

\newcommand{\herm}{\mathsf{H}}			

\newtheorem{lemma}{Lemma}
\newtheorem{definition}{Definition}

\newtheorem{proposition}{Proposition}
\newtheorem{assumption}{Assumption}
\newtheorem{remark}{Remark}
\newtheorem{Cor}{Corollary}

\newcommand{\signal}[1]{{\boldsymbol{#1}}}
\newcommand{\Natural}{{\mathbb N}}

\def\baselinestretch{.99}

\begin{document}
	
	\title{Characterization of the weak Pareto boundary of resource allocation problems in wireless networks -- Implications to cell-less systems
		\thanks{R.~L.~G.~Cavalcante, L.~Miretti, and S.~Sta\'nczak acknowledge the financial support by the Federal Ministry of Education and Research of Germany in the programme of “Souverän. Digital. Vernetzt.” Joint project 6G-RIC, project identification number: 16KISK020K and 16KISK030.}
	}
	
	\author{
		\IEEEauthorblockN{Renato L. G. Cavalcante\textdagger, Lorenzo Miretti\textdagger\textdaggerdbl, and Sławomir Sta\'nczak\textdagger\textdaggerdbl}
		\IEEEauthorblockA{\emph{\textdagger Fraunhofer Institute for Telecommunications, Heinrich-Hertz-Institut, Berlin, Germany} \\
			\emph{\textdaggerdbl Technical University of Berlin, Berlin, Germany}\\
			\{renato.cavalcante,lorenzo.miretti,slawomir.stanczak\}@hhi.fraunhofer.de}
	}
	
	\maketitle
	
	\begin{abstract} We establish necessary and sufficient conditions for a  network configuration to provide utilities that are both fair and efficient in a well-defined sense.  To cover as many applications as possible with a unified framework, we consider utilities defined in an axiomatic way, and the constraints imposed on the feasible network configurations are expressed with a single inequality involving a monotone norm. In this setting, we prove that a necessary and sufficient condition to obtain network configurations that are efficient in the weak Pareto sense is to select configurations attaining equality in the monotone norm constraint. Furthermore, for a given configuration satisfying this equality, we characterize a criterion for which the configuration can be considered fair for the active links. We illustrate potential implications of the theoretical findings by presenting, for the first time, a simple parametrization based on power vectors of achievable rate regions in modern cell-less systems subject to practical impairments.
	\end{abstract}
	

	\section{Introduction}

	In typical formulations of resource allocation problems in wireless networks, each user has its own {\it utility function}, which is a mathematical description of the user's performance [e.g., the signal-to-interference-noise ratio (SINR) or achievable rate] for a given feasible network configuration (e.g., the choice of transmit power and beamformers). In most applications of practical interest, users compete for wireless resources, so the utilities are coupled because changes in the network configuration to improve the utility of one user in general degrades the utility of others. As a result, the utilities cannot be chosen independently, and resource allocation algorithms are inherently limited to choosing points in the so-called achievable (utility) region, which is the set of joint utilities that can be obtained with a valid network configuration \cite{emil2014multiobjective,bjornson2013resource,larsson2011pareto,shindoh2020}. 
	
	To select points in the achievable region, resource allocation algorithms typically solve  optimization problems providing utilities that are both {\it efficient} and {\it fair} in some well-defined sense. In particular, we consider in this study the notion of weak Pareto efficiency, which, in the wireless settings described later, corresponds to the set of utilities, hereafter called weak Pareto region or boundary, for which there is no valid network configuration that strictly increases the utilities of all users \cite{emil2014multiobjective}. This notion is widely used in multi-user information theory to define the boundary of capacity regions, and a standard result is that points on the weak Pareto boundary can be obtained by solving weighted max-min utility optimization problems. From this optimization viewpoint, the weights assigned to the utilities establish a criterion of fairness. For example, in standard power control problems with the rates of users considered as the utilities, we obtain with uniform weighting fair solutions in the sense that the worst rate among all users in the network is maximized, which is a specially useful criterion of fairness in systems aiming at providing uniformly good service to all users, as in the case of cell-less networks \cite{ngo2017cell,schubert2019,demir2021}. 
	
	In light of the above correspondence between points on the weak Pareto boundary and optimal utilities of weighted max-min optimization problems, a great deal of effort has been devoted to constructing simple solvers for the optimization problems (see, for instance, \cite{schubert2004multiuser,bashar2019uplink}). As it is well known in the mathematical literature, particular solutions to many optimization problems can be easily computed by solving fixed point problems for which remarkably simple iterative algorithms exist, and this fact has been widely exploited in the wireless domain \cite{demir2021,nuzman07,tan2014wireless,zheng2016,renatomaxmin,renato2016maxmin,renato2019,miretti22globecom}. From this viewpoint of fixed point theory, the weights defining the fairness criterion can be seen as parameters of a mapping that has as its fixed point a network configuration providing utilities at a specific point on the weak Pareto boundary. 
	
	Building upon these results on fixed point theory, we study the weak Pareto boundary and notions of fairness for a fairly large class of resource allocation problems. In more detail, we start by revisiting a general framework for resource allocation \cite{nuzman07,renato2019} that has been recently used to solve open problems in massive MIMO and cell-less systems \cite{miretti22vtc,miretti22globecom}. In the general setting under consideration, the utilities are described based on the axiomatic framework of standard interference mappings (see \cite{yates95,shindoh2020}, and Assumption~\ref{assumption.restrictions} introduced later), and the feasible network configurations are expressed as an inequality involving monotone norms, and we recall that constraints of this type can represent a fairly rich class of constraints widely used in wireless systems  \cite[Proposition~2]{renatomaxmin}\cite{nuzman07,renato2019,miretti22vtc}. This connection to monotone norms and the above-mentioned results on fixed point theory enable us to establish the following simple parametrization of the weak Pareto boundary: the utilities are on this boundary if and only if the corresponding network configuration achieves equality in the inequality constraint involving the monotone norm (see Corollary~\ref{cor.main}). In the process of proving this result, we also characterize a criterion for which an arbitrary configuration attaining this equality can be considered fair for the active links (Proposition~\ref{proposition.weak_par}). 
	
	\subsection*{Implications and relations to existing studies}
	The above characterizations are not only of great theoretical interest, but they also have important practical implications. For instance, in power control problems with per user power constraints, if we stack the power of each user in a vector, the power constraint is typically expressed as an inequality involving the standard $l_\infty$ norm, which is a particular instance of a monotone norm. In this case, if at least one user transmits with maximum power, we obtain utilities on the weak Pareto boundary, so the resulting network configuration is not only fair for the active users, in a precise sense we describe later in Proposition~\ref{proposition.weak_par}, but also efficient. The remarkable fact to emphasize here is that power vectors attaining equality in the monotone norm constraint give a \textit{necessary and sufficient condition for operation on the weak Pareto boundary}, and power vectors with this property can be easily selected without solving any resource allocation problem. To date, to the best of our knowledge, parametrizations of this type have been reported in the wireless literature only for very particular settings, such as the two user MISO interference channel \cite{larsson2011pareto}, the MISO broadcast channel with a sum power constraint \cite[Sect.~4.1.4]{brehmer2012utility}, and coordinated multi-cell networks under a network-wide sum power constraint \cite[Sect. 3.2.3]{bjornson2013resource}. The parametrizations in \cite{brehmer2012utility,bjornson2013resource}
	are the most related to our study because they are both based on a power allocation vector. However, they both rely on arguments based on deterministic channels (without fading) and to the $l_1$ norm constraint. Therefore, they do not apply to the study of modern massive MIMO and cell-less networks, where relevant impairments such as fading and imperfect CSI must be taken into account. Furthermore, they do not apply to different power constraints such as those based on the $l_{\infty}$ norm, which may be more suitable for uplink operation. Our general proof lifts these restrictions, so our theoretical findings can be directly applied to nontrivial and timely wireless resource allocation problems in cell-less systems, for instance, involving the joint optimization of distributed uplink combiners and transmit powers, as shown in Section~\ref{sect.cell_less}.

	\section{Preliminaries}
	\label{sect.pre}
	In this section we establish notation and the basic definitions used in the proof of the main results. In more detail, the sets of nonnegative and positive reals are denoted by, respectively, $\real_+:=[0,\infty[$ and $\real_{++}:=~]0,\infty[$. We use the convention that the $k$th coordinate of a vector $\signal{x}\in\real^K$ is denoted by $x_k$ and $K\in\Natural$ (i.e., we only consider finite-dimensional settings). Inequalities involving vectors should be understood coordinate-wise, so, for example, $(\forall \signal{x}\in\real_+^K)(\forall \signal{y}\in\real_+^K)(\forall k\in\{1,\ldots,K\})~x_k\le y_k \Leftrightarrow \signal{x}\le\signal{y} $. A norm $\|\cdot\|$ in $\real^K$ is said to be \emph{monotone} if $(\forall \signal{x}\in\real_+^K)(\forall \signal{y}\in\real_+^K)~ \signal{x}\le\signal{y}\Rightarrow \|\signal{x}\|\le\|\signal{y}\|$. 
	We say that a sequence $(\signal{x}_n)_{n\in\Natural}\subset\real^N$ converges  to $\signal{x}^\star$ if $\lim_{n\to\infty}\|\signal{x}_n-\signal{x}^\star\|=0$ for some (and hence for every) norm $\|\cdot\|$ in $\real^K$. A set $C\subset\real^K_+$ is said to be downward comprehensive (on $\real^K_+$) if $(\forall\signal{x}\in\real_+^K)(\forall\signal{y}\in C)~\signal{x}\le\signal{y}\Rightarrow\signal{x}\in C$.
	
	The following class of functions plays a crucial role in the results in the next sections: 
	
	\begin{definition}
		\label{definition.mappings} \cite{yates95} A function $f:\real^K_+\to\real_{++}$ is said to be a \emph{standard interference function} if the following properties hold:  \par
		
		[monotonicity] $(\forall \signal{x}\in\real^K_{+})(\forall \signal{y}\in\real^N_{+}) ~ \signal{x}\ge\signal{y} \Rightarrow f(\signal{x})\ge f(\signal{y})$; and \par 
		
		[scalability] $(\forall \signal{x}\in\real^K_+)$ $(\forall \alpha>1)$  $\alpha {f}(\signal{x})>f(\alpha\signal{x})$. \par 
		
		Likewise, a mapping $T:\real_+^K\to\real_{++}^K:\signal{x}\mapsto[f_1(\signal{x}),\ldots,f_N(\signal{x})]$ is said to be a standard interference (SI) mapping if each coordinate function $f_k:\real^K_+\to\real_{++}$ ($k=1,\ldots,K$) is a standard interference function.
	\end{definition}

	We recall that the set of (positive) concave functions $f:\real^K_+\to\real_{++}$ is a proper subset of standard interference functions \cite[Proposition~1]{renato2016}. The set of fixed points of a mapping $T:\real_+^K\to\real_+^K$ is denoted by $\mathrm{Fix}(T):=\{\signal{x}\in\real_+^K~|~T(\signal{x})=\signal{x}\}$.  
	
	In this study, we derive properties of utility optimization problems that have the objective of allocating wireless resources among transceivers in an efficient way, and here we are particularly interested in the following well-established notion of efficiency:

	\begin{definition} \cite{emil2014multiobjective}
		\label{def.def_paretto}
		Given $K\in\Natural$ (utility) functions $(u_k:\real_+^K\to\real)_{k=1,\ldots,K}$, a nonempty set $C\subset\real_+^K$, and a vector $\signal{p}^\star\in C$, we say that the utilities $(u_1(\signal{p}^\star),\ldots,u_K(\signal{p}^\star))$ are on the weak Pareto boundary (of the set of utility values that are simultaneously attainable) under the constraint $C$ if there is no vector $\signal{p}^\prime\in C$ such that $(\forall k\in\{1,\ldots,K\})~u_k(\signal{p}^\prime)>u_k(\signal{p}^\star)$.
	\end{definition}

	\section{The weak Pareto boundary of resource allocation problems in wireless networks}
	
	\label{sect.wpar}

	\subsection{Problem statement and overview of the main results}
	
	Denote by $u_k:\real^K_{+}\to\real_{+}$ the utility function of user $k$ in a system with $K$ users, and by $\mathcal{K}:=\{1,\ldots,K\}$ the set of users. In the applications that can be modeled using the general framework described here, users are coupled by the argument of these functions; i.e., the same vector $\signal{p}$ is used in every function to characterize the overall utility $(u_1(\signal{p}),\ldots,u_K(\signal{p}))$ of the system. For example, in a wireless network, each value $u_k(\signal{p})$ could correspond to the SINR or rate that user $k$ can achieve if transmitters use power $\signal{p}=(p_1,\ldots,p_k)$, where $p_k\ge 0$ is the transmit power of user $k$. For given $\bar{p}\in\real_+$, let $C:=\{\signal{p}\in\real_+^K~|~\|\signal{p}\|\le\bar{p}\}$ be the constraint imposed on the vector $\signal{p}$, where, as in \cite{nuzman07}, $\|\cdot\|$ is a monotone norm on $\real^K$. We recall that the constraint $\|\signal{p}\| \le \bar{p}$ is fairly general: any set $S\subset\real^K_+$ with nonempty interior that is also convex, compact, and downward-comprehensive can be equivalently written as $S=C:=\{\signal{p}\in\real_+^K~|~\|\signal{p}\|\le 1\}$, where the monotone norm \linebreak[4] $\|\cdot\|$ is the Minkowski functional of $\mathrm{conv}(-S\cup S),$ and where $\mathrm{conv(\cdot)}$ denotes the convex hull of a set: \cite[Proposition 2]{renatomaxmin}: $(\forall\signal{p}\in\real^K)~\|\signal{p}\|:=\inf\{\gamma>0~|~(1/\gamma)\signal{p}\in \mathrm{conv}(-S\cup S)\}$. To keep the class of resource allocation problems tractable and yet fairly general, we impose the following common assumption on the utility functions \cite{shindoh2020,yates95}, and later in Sect.~\ref{sect.cell_less} we show that it is valid in optimal distributed combining (receive beamforming) design and  power allocation in cell-less networks:
	\begin{assumption}
		\label{assumption.restrictions}
		For every $k\in\mathcal{K}$, there exists a \emph{continuous} standard interference function $f_k:\real_+^K\to\real_{++}$ such that the utility $u_k$  can be equivalently written as  
		\begin{align}
			\label{eq.util}
			u_k:\real_+^K\to\real_+:\signal{p}\mapsto\dfrac{p_k}{f_k(\signal{p})}.
		\end{align}

	\end{assumption}

	With the above definitions, \emph{the main objective of this study is to characterize the set of vectors $\signal{p}$ that provides utilities $(u_1(\signal{p}), \ldots, u_K(\signal{p}))$ on the weak Pareto boundary under the constraint $C$.} More precisely, in Propositions \ref{proposition.boundary} and \ref{proposition.maxmin} in the next subsection we revisit the fact that the set $\mathcal{B}$ of utilities on the weak Pareto boundary can be obtained by solving weighted max-min problems with existing fixed point algorithms. This result is then used to prove in Corollary~\ref{cor.main} that the set $\mathcal{B}$ can be expressed as $\mathcal{B}=\{(u_1(\signal{p}),\ldots,u_K(\signal{p}))\in\real_+^K~|~\signal{p}\in\real_+^K\text{ and }\|\signal{p}\|=\bar{p}\}$. In particular, this interesting characterization shows that efficient utilities in the sense of Definition~\ref{definition.mappings} can be obtained by simply selecting any vector $\signal{p}\in\real^K_+$ satisfying $\|\signal{p}\|=\bar{p}$. In addition, given a vector $\signal{p}$ with this property, we show in Proposition~\ref{proposition.weak_par}  a criterion that characterizes the resulting utilities as fair for the active links (i.e., for a reduced network constructed by removing links $k\in\mathcal{K}$ such that $p_k=0$).

	\subsection{Characterization of the weak Pareto boundary}
	\label{sect.max_min}
	
	To prove the main results in this study, we use extensively the following implications of Assumption~\ref{assumption.restrictions}.
	
	\begin{lemma}
		\label{lemma.properties_u}
		Utility functions $(u_k)_{k\in\mathcal{K}}$ satisfying Assumption~\ref{assumption.restrictions} have the following properties:
		\begin{itemize}
			\item[(i)] All functions $(u_k)_{k\in\mathcal{K}}$ are continuous.
			\item[(ii)] $(\forall k\in\mathcal{K})(\forall\signal{p}\in\real_{++}^K)~u_k(\signal{p})>0$.
			\item[(iii)]  $(\forall k\in\mathcal{K})(\forall \signal{p}\in\real_+^K)~ u_k(\signal{p})=0\Leftrightarrow p_k=0$.
			\item[(iv)] $(\forall k\in\mathcal{K})(\forall \signal{p}\in\real_{+}^{K})(\forall \alpha>1)~ p_k > 0\Rightarrow u_k(\alpha \signal{p}) > u_k(\signal{p})$.
			\item[(v)] $(\forall k\in\mathcal{K})(\forall \signal{p}\in\real_{+}^{K})(\forall \signal{x}\in\real_{+}^{K})~ {p}_k={x}_k\text{ and }\signal{p}\ge\signal{x} \Rightarrow u_k(\signal{p})\le u_k(\signal{x})$.
		\end{itemize}

		\end{lemma}
		\begin{proof}
			(i)-(iii) Immediate from positivity and continuity of $(f_k)_{k\in\mathcal{K}}$ on $\real_{+}^K$.
			
			(iv) Use positivity and scalability of standard interference functions in \refeq{eq.util} to deduce $(\forall k\in\mathcal{K})(\forall \signal{p}\in\real_+^K)(\forall\alpha>1)$
			\begin{align*}
				p_k > 0 \Rightarrow u_k(\alpha \signal{p})=\dfrac{\alpha p_{k}}{f_k(\alpha \signal{p})} > \dfrac{\alpha~p_k}{\alpha~f_k(\signal{p})} = u_k(\signal{p}).
			\end{align*}
			
			(v) Assume that the tuple  $(\signal{p},\signal{x})\in\real_+^K\times\real_+^K$ satisfies $p_k={x}_k$ and $\signal{p}\ge\signal{x}$ for some $k\in\mathcal{K}$. Monotonicity of standard interference functions and \refeq{eq.util} yield
			\begin{align*}
				u_k(\signal{p})=\dfrac{p_k}{f_k(\signal{p})} \le \dfrac{p_k}{f_k(\signal{x})} = \dfrac{x_k}{f_k(\signal{x})} = u_k(\signal{x}),
			\end{align*}
			and the proof is complete.
		\end{proof}

		Now, consider the following weighted max-min utility optimization problem:
		
		\begin{align}
			\label{problem.maxmin}
			\begin{array}{rl}
				\text{maximize}_{\signal{p}\in\real_{+}^K} & \min_{k\in\mathcal{K}} \omega_k^{-1} u_k(\signal{p}) \\
				\text{subject to} & \|\signal{p}\| \le \bar{p},
			\end{array}
		\end{align}
		where, as in the previous subsection, $\|\cdot\|$ is a monotone norm on $\real^K$, $\bar{p}>0$ is the maximum (power) budget, and $(\omega_1,\ldots,\omega_K)\in\real_{++}^K$ are weights/priorities assigned to the utility of user $k$. In particular, these weights define the criterion of fairness in the optimization problem, and a well-known result is that points with all positive utilities on the weak Pareto boundary can be obtained by changing the weights. Since previous studies often use slightly different assumptions and notation, we prove this standard result for the settings under consideration in the next proposition.
		
		\begin{proposition}
			\label{proposition.boundary}
			Given a monotone norm $\|\cdot\|$ on $\real_+^K$, a scalar $\bar{p}>0$, a vector $\signal{p}^\star\in C:=\{\signal{p}~\in\real_+^K|~\|\signal{p}\|\le\bar{p}\}$, and the utilities $(u_k)_{k\in\mathcal{K}}$ in problem~\refeq{problem.maxmin}, assume that $(\forall k\in\mathcal{K})~u_k(\signal{p}^\star)>0$. If the utilities $(u_1(\signal{p}^\star),\ldots,u_K(\signal{p}^\star))$ are on the weak Pareto boundary under the constraint $C$, then $\signal{p}^\star$ is a solution to problem~\refeq{problem.maxmin} with weights given by
			\begin{align}
				\label{eq.optw}
				(\forall k\in\mathcal{K})\quad\omega_k=u_k(\signal{p}^\star).
			\end{align}
		\end{proposition}
		\begin{proof}
			The proof is obtained with a simple contradiction. If the utilities $(u_1(\signal{p}^\star),\ldots,u_K(\signal{p}^\star))$ are on the weak Pareto boundary under the constraint $C$, we have  $\signal{p}^\star\in C$, so $\signal{p}^\star$ is feasible to problem~\refeq{problem.maxmin}. Therefore, if $\signal{p}^\star$ is not a solution to problem~\refeq{problem.maxmin} with the weights in \refeq{eq.optw}, there exists $\signal{p}^\prime\in C$ such that 
			\begin{multline*}
				\min_{k\in\mathcal{K}}~\dfrac{1}{u_k(\signal{p}^\star)} u_k(\signal{p}^\prime) > \min_{k\in \mathcal{K}} \dfrac{1}{u_k(\signal{p}^\star)} u_k(\signal{p}^\star) = 1 \\ \Leftrightarrow (\forall k\in\mathcal{K})~u_k(\signal{p}^\prime) > u_k(\signal{p}^\star),
			\end{multline*}
			and thus we contradict that $(u_1(\signal{p}^\star),\ldots,u_K(\signal{p}^\star))$ is on the weak Pareto boundary under the constraint $C$.
		\end{proof}

		We now proceed to relate problem~\refeq{problem.maxmin} with the fixed point problems discussed in \cite{renato2019,nuzman07}. Similar relations have also been shown in \cite{tan2014wireless,zheng2016}, but a careful look at  Lemma~\ref{lemma.properties_u} shows that the competitiveness assumption in \cite[Assumption~1]{tan2014wireless}\cite[Assumption~1]{zheng2016} is not required in our framework. As a result, we are able to cover important models such as the scalar Gaussian multiple access channel or the scalar Gaussian broadcast channel under capacity achieving coding schemes based on successive interference cancellation. By lifting the competitiveness  assumption, the solution to problem \refeq{problem.maxmin} is not necessarily unique, and existing results characterizing the weak Pareto boundary in very particular settings, such as those in \cite[Corollaries~4.1.8 and 4.1.12]{brehmer2012utility}, do not necessarily hold.
		
		Writing problem \refeq{problem.maxmin} in its epigraph form, we obtain:
		\begin{align}
			\label{problem.maxmin2}
			\begin{array}{rl}
				\text{maximize}_{(c,\signal{p})\in\real_+\times\real_+^K} &  c  \\
				\text{subject to} & \|\signal{p}\| \le \bar{p} \\ 
				& (\forall k\in\mathcal{K})\quad u_k(\signal{p})\ge  c~\omega_k.
			\end{array}
		\end{align}
		
		\begin{remark} 
			\label{remark.equivalence}
			Problems \refeq{problem.maxmin} and \refeq{problem.maxmin2} are equivalent in the following sense. If $(c^\star,\signal{p}^\star)$ solves problem~\refeq{problem.maxmin2}, then $\signal{p}^\star$ is a solution to problem~\refeq{problem.maxmin}. Conversely, if $\signal{p}^\star$ is a solution to problem~\refeq{problem.maxmin}, then $\left(c^\star:=\min_{k\in\mathcal{K}} \omega_k^{-1} u_k(\signal{p}^\star),\signal{p}^\star\right)$ is a solution to problem~\refeq{problem.maxmin2}. Furthermore, Lemma~\ref{lemma.properties_u}(ii) and $\bar{p}>0$ imply $c^\star\ge \min_{k\in\mathcal{K}} u_k(\signal{v})/\omega_k>0$, where $\signal{v}=({\bar{p}}/{\|\signal{1}_K\|})\signal{1}_K$ and $\signal{1}_K\in\real_+^K$ is the vector of ones. As a result, we also have $(\forall k\in\mathcal{K})~{u}_k(\signal{p}^\star)>0$, and thus $\signal{p}^\star\in\real_{++}^K$ as a consequence of  Lemma~\ref{lemma.properties_u}(iii).
		\end{remark}

		The next simple lemma shows that the inequalities $(\forall k\in\mathcal{K})\quad u_k(\signal{p})\ge  c~\omega_k$ in problem \refeq{problem.maxmin2} are equivalent to a vector inequality involving standard interference mappings.
		
		\begin{lemma}
			\label{lemma.constraint}
			Fix $\signal{p}=[p_1,\ldots,p_K]\in\real^K_+$ and $c\ge0$. Then the inequalities $(\forall k\in\mathcal{K})~u_k(\signal{p})\ge c~\omega_k$ in the constraints of problem~\refeq{problem.maxmin2} are satisfied if and only if $(\forall k\in\mathcal{K})~p_k\ge c\omega_k f_k(\signal{p})$, or, in vector form, $\signal{p} \ge cT(\signal{p})$, where
			\begin{equation}
				\label{eq.mappingT}
				T:\real^K_{+}\to \real^K_{++}:\signal{p}\mapsto \left[\begin{matrix} \omega_1  f_1(\signal{p}) \\ \vdots \\ \omega_K f_K(\signal{p}) \end{matrix}\right],
			\end{equation}
			and $(f_k)_{k\in\mathcal{K}}$ and $(u_k)_{k\in\mathcal{K}}$ are related according to Assumption~\ref{assumption.restrictions}. In particular, $(\forall k\in\mathcal{K})~u_k(\signal{p})=  c~\omega_k$ if and only if $\signal{p}=cT(\signal{p})$.
		\end{lemma}
		\begin{proof}
			The case $c=0$ is trivial because the functions $(u_k)_{k\in\mathcal{K}}$ are nonnegative, so we only consider the case $c>0$. Without any loss of generality, for every $k\in\mathcal{K}$, we can further assume that $p_k>0$ and $u_k(\signal{p})>0$ because otherwise the inequality $u_k(\signal{p})\ge c\omega_k$ or the inequality $p_k\ge c\omega_k f_k(\signal{p})$, or both, cannot hold in light of Lemma~\ref{lemma.properties_u}(iii), positivity of $f_k$, and positivity of the weight $\omega_k$. As a result, 
			\begin{multline}
				\label{eq.ineqa}
				(\forall k\in\mathcal{K})~
				u_k(\signal{p})\ge c \omega_k \Leftrightarrow 1 \ge \dfrac{c \omega_k}{u_k(\signal{p})} \\ \Leftrightarrow p_k \ge \dfrac{p_kc \omega_k}{u_j(\signal{p})} \Leftrightarrow p_k\ge c\omega_k f_k(\signal{p})>0,
			\end{multline} 
			and the proof of the first part of the lemma is complete. Replacing the weak inequalities in \refeq{eq.ineqa} with equalities, we conclude that $(\forall k\in\mathcal{K})~p_k = c\omega_k f_k(\signal{p}) \Leftrightarrow \signal{p}=cT(\signal{p})$, which completes the proof.
		\end{proof}

		Replacing the inequalities in the constraints $(\forall k\in\mathcal{K})\quad u_k(\signal{p})\ge  c~\omega_k$ in problem~\refeq{problem.maxmin2} with equalities, and using Lemma~\ref{lemma.constraint}, we obtain the following optimization problem, which we later show to be able to obtain a particular solution to problem~\refeq{problem.maxmin2} [and hence to problem~\refeq{problem.maxmin}]:
		
		\begin{align}
			\label{problem.maxmin3}
			\begin{array}{rl}
				\text{maximize}_{(c,\signal{p})\in\real_+\times\real_+^K} & ~ c  \\
				\text{subject to} & \|\signal{p}\| \le \bar{p} \\ 
				& \signal{p}=cT(\signal{p}),
			\end{array}
		\end{align}
		where $T:\real^K_+\to\real_{++}^K$ is the mapping in \refeq{eq.mappingT}. The next proposition relates all the optimization problems discussed above.
		
		\begin{proposition} \label{proposition.maxmin} Assume that problems \refeq{problem.maxmin}, \refeq{problem.maxmin2}, and \refeq{problem.maxmin3} use the same monotone norm $\|\cdot\|$ and power budget $\bar{p}>0$. Then each of the following holds:
			
			(i) If $(c^\star, \signal{p}^\star)$ solves problem \refeq{problem.maxmin2}, then there exists $\signal{p}^\prime\in\real_{++}^K$  such that   $\signal{p}^\prime\le\signal{p}^\star$ and $(c^\star,\signal{p}^\prime)$ is a common solution to problems  \refeq{problem.maxmin2} and \refeq{problem.maxmin3}.
			
			(ii) Problem~\refeq{problem.maxmin3} has a unique solution given by $(\bar{p}/\|T(\signal{p}^\star)\|,~\signal{p}^\star)$, where $\signal{p}^\star\in\real_{++}^K$ is the unique fixed point of the mapping
			\begin{align}
				\label{eq.nmapping}
				\tilde{T}:\real^K_{+}\to\real^K_{++}:\signal{p}\mapsto \dfrac{\bar{p}}{\|T(\signal{p})\| } T(\signal{p}).
			\end{align}
			Furthermore, for any $\signal{p}_1\in\real_{+}^K$, the vector $\signal{p}^\star$ is also the limit of the sequence $(\signal{p}_n)_{n\in\Natural}$ generated via $\signal{p}_{n+1} = \tilde{T}(\signal{p}_{n})$.

			(iii) If $(c^\star, \signal{p}^\star)$ solves problem \refeq{problem.maxmin3}, then $(c^\star, \signal{p}^\star)$ also solves problem \refeq{problem.maxmin2}, and hence $\signal{p}^\star$ is a solution to problem \refeq{problem.maxmin} as a consequence of Remark~\ref{remark.equivalence}.

		\end{proposition}
		\begin{proof}
			
			(i) Taking into account Lemma~\ref{lemma.constraint}, we verify that the result is immediate if $\signal{p}^\star = c^\star T(\signal{p}^\star)$, so we only consider the case $\signal{p}^\star \ge c^\star T(\signal{p}^\star)$ with strict inequality in at least one coordinate. Since the optimal utility $c^\star>0$ as a consequence of Remark~\ref{remark.equivalence},  the mapping defined by $\mathcal{I}:\real_{+}^K\to\real_{++}^K:\signal{p}\mapsto c^\star T(\signal{p})$ is a standard interference mapping. It follows from \cite[Fact~4]{renato2016} and $\signal{p}^\star \ge c^\star T(\signal{p}^\star)=\mathcal{I}(\signal{p}^\star)$  that (i) $\mathrm{Fix}(\mathcal{I})\neq\emptyset$, (ii) the fixed point $\signal{p}^\prime \in \mathrm{Fix}(\mathcal{I})$ of the mapping $\mathcal{I}$ is unique, (iii) and $\signal{p}^\prime \le \signal{p}^\star$. Furthermore, $\signal{p}^\prime\neq\signal{p}^\star$ because at least one inequality in the constraint $c^\star T(\signal{p}^\star)\le \signal{p}^\star\notin \mathrm{Fix}(\mathcal{I})$ is assumed to be strict. Since the norm $\|\cdot\|$ is monotone and $\|\signal{p}^\star\|\le \bar{p}$ [because $(c^\star, \signal{p}^\star)$ solves problem \refeq{problem.maxmin2} by assumption], the inequality $\signal{p}^\prime \le \signal{p}^\star$ implies $\|\signal{p}^\prime\|  \le \|\signal{p}^\star\|\le \bar{p}$. As a result, $(c^\star,\signal{p}^\prime)$ satisfies all constraints in problem \refeq{problem.maxmin2}, so this tuple is also a solution to problem \refeq{problem.maxmin2} because the optimal objective value $c^\star$ is unchanged. Furthermore, since $\signal{p}^\prime=c^\star T(\signal{p}^\prime)$ because $\signal{p}^\prime\in\mathrm{Fix}(\mathcal{I})$ by definition,  the tuple   $(c^\star,\signal{p}^\prime)$ is feasible to both problem~\refeq{problem.maxmin2} and problem~\refeq{problem.maxmin3}. Recalling that the feasible set of problem \refeq{problem.maxmin2} contains all points of the feasible set of problem \refeq{problem.maxmin3}, and both problems have the same objective function, we conclude that $(c^\star,\signal{p}^\prime)$ has to be a solution to problem \refeq{problem.maxmin3} because the optimal objective of problem \refeq{problem.maxmin3} is upper bounded by the optimal objective $c^\star$ of problem \refeq{problem.maxmin2}, and we have just shown that, with the feasible tuple $(c^\star,\signal{p}^\prime)$, the optimal objective $c^\star$ is attained for both problem~\refeq{problem.maxmin2} and problem~\refeq{problem.maxmin3}.
			
			(ii) Immediate from \cite[Fact~5(i)-(ii)]{renato2019}\cite{nuzman07}.
			
			(iii) Part (ii) of the proposition shows that problem~(\ref{problem.maxmin3}) has a unique solution that we denote by $(c^\star, \signal{p}^\star)$. From the Weierstrass extreme value theorem, we know that  problem~\refeq{problem.maxmin} has at least one solution because its constraint is a compact set and the cost function is continuous as an implication of Lemma~\ref{lemma.properties_u}(i). Therefore, it follows from  Remark~\ref{remark.equivalence} that problem~\refeq{problem.maxmin2} also has a solution. Part (i) of the proof now shows that problems \refeq{problem.maxmin2} and \refeq{problem.maxmin3} have a common solution, so it must be $(c^\star, \signal{p}^\star)$ because the solution to problem~\refeq{problem.maxmin3} is unique.\end{proof}	
	
	Remark~\ref{remark.equivalence} and the next proposition show that any positive vector $\signal{p}^\star$ satisfying $\|\signal{p}^\star\|=\bar{p}$ is a solution to problem~\refeq{problem.maxmin} for a specific choice of weights, which are the parameters defining the criterion of fairness. As a result, such a vector $\signal{p}^\star$ not only provides fairness in a well-defined sense, but it is also provides efficient utilities in the sense of Definition~\ref{def.def_paretto}.

	\begin{proposition}
		\label{proposition.weak_par}
		Let $\signal{p}^\star=[p_1^\star,\ldots,p_K^\star]\in\real_{++}^K$ satisfy $\|\signal{p}^\star\|=\bar{p}$ for a given monotone norm $\|\cdot\|$ and $\bar{p}>0$. Then $\signal{p}^\star$ solves problem \refeq{problem.maxmin} with the same norm $\|\cdot\|$ and power budget $\bar{p}$ if the weights are set to
		\begin{align}
			\label{eq.weights}
			(\forall k\in\mathcal{K})\quad \omega_k = \dfrac{p^\star_k}{f_k(\signal{p}^\star)}={u_k(\signal{p}^\star)}>0,
		\end{align}
		where the functions $(f_k)_{k\in\mathcal{K}}$ are defined in Assumption~\ref{assumption.restrictions}. Furthermore, $(c^\star:=\bar{p}/{\|T(\signal{p}^\star)\|}, \signal{p}^\star)$ is a common solution to problems~\refeq{problem.maxmin2} and \refeq{problem.maxmin3}, where $T$ is the mapping in \refeq{eq.mappingT} with the weights in \refeq{eq.weights}.
	\end{proposition}
	\begin{proof}
		From the definition of $T$, we verify that $\signal{p}^\star$ satisfies  $T(\signal{p}^\star)=\signal{p}^\star$, and thus we also have $\|T(\signal{p}^\star)\|=\|\signal{p}^\star\|=\bar{p}>0$. As a result, $\signal{p}^\star = T(\signal{p}^\star) = ({\bar{p}}/{\|T(\signal{p}^\star)\|}) T(\signal{p}^\star)$,
		which shows that $\signal{p}^\star$ is the fixed point of the mapping in \refeq{eq.nmapping}. The desired result now follows from Proposition~\ref{proposition.maxmin}(ii)-(iii). 
	\end{proof}

	One possible limitation of Proposition~\ref{proposition.weak_par} is that it takes into  account  only positive vectors $\signal{p}\in\real_{++}^K$. For nonnegative vectors, we could in principle consider a reduced network where we remove links $k\in\mathcal{K}$ such that  $p_k=0$. However,  nonnegative vectors should not be discarded  because they are useful to represent wireless systems where some users do not transmit because of scheduling decisions. In  Proposition~\ref{proposition.main_ch4} below, we not only address this limitation, but we also prove the converse of an implication of the first statement in Proposition~\ref{proposition.weak_par}.
	
	\begin{proposition}
		\label{proposition.main_ch4}
		Let $\|\cdot\|$ be an arbitrary monotone norm in $\real^K$, and fix the scalar $\bar{p}>0$ arbitrarily. Then the utilities  $(u_1(\signal{p}^\star),\ldots,u_K(\signal{p}^\star))$ corresponding to a vector $\signal{p}^\star\in C:=\{\signal{p}~\in\real_+^K|~\|\signal{p}\|\le\bar{p}\}$ are on the weak Pareto boundary under the constraint $C$ if and only if $\|\signal{p}^\star\|=\bar{p}$. 
	\end{proposition}
	\begin{proof}
		Let $\signal{p}^\star\in C$ be any vector satisfying $\|\signal{p}^\star\|=\bar{p}$.  For the sake of contradiction, suppose that the utilities $(u_1(\signal{p}^\star),\ldots,u_K(\signal{p}^\star))$ are not on the weak pareto Boundary under the constraint $C$. As a result there exists $\signal{p}^\prime\in C$ such that 
		\begin{align}
			\label{eq.contr}
			\min_{k\in\mathcal{K}} (u_k(\signal{p}^\prime)-u_k(\signal{p}^\star))>0.
		\end{align}
		Let   $(\signal{b}_n)_{n\in\Natural}$ be any sequence in $\real_{++}^K$  converging in norm to $\signal{p}^\star\in C$, and construct the normalized sequence $(\signal{p}_n:=(\bar{p}/\|\signal{b}_n\|)\signal{b}_n)_{n\in\Natural}$, which implies $(\forall n\in\Natural) \|\signal{p}_n\|=\bar{p}$ and $\signal{p}_n\in C\cap\real_{++}^N$. Proposition~\ref{proposition.weak_par} shows that, for every $n\in\Natural$, the positive vector $\signal{p}_n\in C$ solves a max-min problem of the type in \refeq{problem.maxmin}, so the utilities  $(u_1(\signal{p}_n),\ldots,u_K(\signal{p}_n))$ are on the weak pareto Boundary under the constraint $C$. Therefore, 
		\begin{align}
			\label{eq.le0}
			(\forall n\in\Natural)~\min_{k\in\mathcal{K}} (u_k(\signal{p}^\prime)-u_k(\signal{p}_n))\le 0.
		\end{align}
		Lemma~\ref{lemma.properties_u}(i) shows that the utility functions $(u_k)_{k\in\mathcal{K}}$ are continuous, so the function 
		\begin{align}
			\label{eq.haux}
			h:\real^K_+\to\real:\signal{x}\mapsto\min_{k\in\mathcal{K}}(u_k(\signal{p}^\prime)-u_k(\signal{x}))
		\end{align} is also continuous. As a result, we obtain the contradiction
		\begin{multline*}
			0\overset{\text{(a)}}{\ge} \lim_{n\to\infty} \min_{k\in\mathcal{K}} (u_k(\signal{p}^\prime)-u_k(\signal{p}_n)) \\ \overset{\text{(b)}}{=}  \min_{k\in\mathcal{K}} (u_k(\signal{p}^\prime)-u_k(\lim_{n\to\infty} \signal{p}_n)) \overset{\text{(c)}}{=}  \min_{k\in\mathcal{K}} (u_k(\signal{p}^\prime)-u_k(\signal{p}^\star)) \overset{\text{(d)}}{>} 0,
		\end{multline*}
		where (a) follows from \refeq{eq.le0}, (b) follows from  continuity of the function $h$ defined in \refeq{eq.haux}, (c) follows from the definition of the sequence $(\signal{p}_n)_{n\in\Natural}$, and (d) follows from \refeq{eq.contr}. This contradiction completes the proof that the utilities $(u_1(\signal{p}^\star),\ldots,u_K(\signal{p}^\star))$ are on the weak pareto Boundary under the constraint $C$ if $\signal{p}^\star\in C$ satisfies $\|\signal{p}^\star\|=\bar{p}$.
		
		Conversely, let $\signal{p}^\star\in C$ be a given vector such that the utilities $(u_k(\signal{p}^\star))_{k\in\mathcal{K}}$ are on the weak pareto Boundary under the constraint $C$. Suppose that $0\le \|\signal{p}^\star\| < \bar{p}$. If $\|\signal{p}^\star\|=0\Leftrightarrow\signal{p}^\star=\signal{0}$, then Lemma~\ref{lemma.properties_u}(iii) shows that any vector in $C\cap\real_{++}^K\neq\emptyset$ increases all utilities, so we contradict that the utilities corresponding to $\signal{p}^\star$ are on the weak Pareto boundary under $C$. We now only have to show a contradiction for $0<\|\signal{p}^\star\|<\bar{p}$. To this end, let $\alpha:=\bar{p}/\|\signal{p}^\star\|>1$ so that $\|\alpha \signal{p}^\star\|=\bar{p}$ and $\alpha\signal{p}^\star\in C$. Denote by $\mathcal{P}:=\{k\in\mathcal{K}~|~p^\star_k>0\}\neq\emptyset$ the set of coordinates of the vector $\signal{p}^\star$ with  positive elements, and by $\mathcal{Z}:=\mathcal{K}\backslash\mathcal{P}$ the possibly empty set of coordinates of the vector $\signal{p}^\star$ with elements taking the value zero.  From Lemma~\ref{lemma.properties_u}(iv) we deduce 
		\begin{align}
			\label{eq.positive}
			(\forall k\in\mathcal{P})~u_k(\signal{p}^\star)<u_k(\alpha\signal{p}^\star). 
		\end{align}
		
		Let $(\signal{x}_n)_{n\in \mathbb{N}}$ be a sequence in $\real_{++}^K$ converging in norm to $\alpha\signal{p}^\star$, and define $(\forall n\in\Natural)~ \signal{y}_n:=({\bar{p}}/{\|\signal{x}_n\|})\signal{x}_n$, which implies $(\forall n\in\Natural)~\|\signal{y}_n\| = \bar{p}$. As a result, recalling that  $\|\alpha\signal{p}^\star\|=\bar{p}>0$, we have  $\lim_{n\to\infty}\signal{y}_n=(\bar{p}/\|\alpha\signal{p}^\star\|)\alpha\signal{p}=\alpha\signal{p}^\star$  and  $(\forall n\in\Natural)~\signal{y}_n\in C\cap\real_{++}^N$. Continuity of the functions $(u_k)_{k\in\mathcal{K}}$ [Lemma~\ref{lemma.properties_u}(i)] yield $(\forall k\in\mathcal{K}) \lim_{n\to\infty}u_k(\signal{y}_n) = u_k(\alpha\signal{p}^\star)$, which, together with \refeq{eq.positive}, implies the existence of $N\in\Natural$ such that 
		\begin{align}
			\label{eq.ineqabca}
			(\forall k\in\mathcal{P})~u_k(\signal{p}^\star)<u_k(\signal{y}_N).
		\end{align}
		Recalling that $\signal{y}_N\in C\cap\real_{++}^N$ by construction, we obtain from Lemma~\ref{lemma.properties_u}(iii) that
		\begin{align}
			\label{eq.ineqabcb}
			(\forall k\in\mathcal{Z})~0=u_k(\signal{p}^\star)<u_k(\signal{y}_N). 
		\end{align}
		Combining \refeq{eq.ineqabca} and \refeq{eq.ineqabcb}, we deduce \linebreak[4] $(\forall k\in\mathcal{K}=\mathcal{P}\cup\mathcal{Z})~u_k(\signal{p}^\star)<u_k(\signal{y}_N),$ which contradicts that the utilities $(u_1(\signal{p}^\star),\ldots,u_K(\signal{p}^\star))$ are on the weak Pareto boundary under the constraint $C$, so we must have $\|\signal{p}^\star\|=\bar{p}$.
	\end{proof}

	The practical implication of Proposition~\ref{proposition.main_ch4} is that, to probe the weak Pareto boundary, we do not need to solve problem~\refeq{problem.maxmin} with different weights. It suffices to choose vectors with maximum norm and compute the corresponding utilities. The next corollary formalizes this statement.
		
	\begin{Cor}
		\label{cor.main}
		Given utility functions $(u_k)_{k\in\mathcal{K}}$ satisfying Assumption~\ref{assumption.restrictions}, denote by $\mathcal{B}\subset \real_+^K$ the set of utilities on the weak Pareto boundary under the constraint $C:=\{\signal{p}\in\real_+^K~|~\|\signal{p}\|\le\bar{p}\}$, where $\|\cdot\|$ is a given monotone norm and $\bar{p}>0$ a given power budget. Then 
		\begin{align*}
			\mathcal{B}=\{(u_1(\signal{p}),\ldots,u_K(\signal{p}))\in\real_+^K~|~\signal{p}\in\real_+^K\text{ and }\|\signal{p}\|=\bar{p}\}.
		\end{align*}
	\end{Cor}
	
	\section{Implications to cell-less systems}
	\label{sect.cell_less}
	\subsection{System model}
	This sections illustrates an application of the theoretical findings to a nontrivial resource allocation problem in modern cell-less networks, which cover cellular networks as a particular case. In particular, we study the uplink performance of a cell-less network as defined in \cite{demir2021}, composed of $L$ access-points (APs) indexed by $\mathcal{L}:=\{1,\ldots,L\}$, each of them equipped with $N$ antennas, and $K$ single-antenna users indexed by $\mathcal{K}:=\{1,\ldots,K\}$.
	We consider ergodic achievable rates measured via the popular \textit{use-and-then-forget} (UatF) lower bound \cite{demir2021}. More precisely, the achievable rate of user $k\in\set{K}$ is given by		

	\begin{equation*}
		R_k^{\mathrm{UL}}(\rvec{v}_k,\vec{p}) \eqdef \log_2(1+\mathrm{SINR}_k^{\mathrm{UL}}(\rvec{v}_k,\vec{p})),
	\end{equation*}
where
	\begin{equation*}
		\mathrm{SINR}_k(\rvec{v}_k,\vec{p}) \eqdef \resizebox{0.69\linewidth}{!}{$\dfrac{p_k|\E[\rvec{h}_k^\herm\rvec{v}_k]|^2}{p_k\Var(\rvec{h}_k^\herm\rvec{v}_k)+\underset{j\neq k}{\sum} p_j\E[|\rvec{h}_j^\herm\rvec{v}_k|^2]+\E[\|\rvec{v}_k\|_2^2]}$},
	\end{equation*}
	$\rvec{h}_{k}$ is a random vector taking values in $\stdset{C}^{NL}$, which models the fading channel between all APs and user~$k$; $\rvec{v}_k$ is a random vector taking values in $\stdset{C}^{NL}$, which models a distributed combining vector applied by all APs to jointly process the signal of user~$k$; $\vec{p} \in \real_+^K$ is a deterministic vector, which models the power allocated to the users; and $\E[\cdot]$ and $\Var(\cdot)$ denote, respectively, the expectation and the variance of random variables. 
	
	Similarly to many studies on cell-less networks, we focus on an instance of the weighted max-min fair utility optimization problem in \eqref{problem.maxmin} with the utility for user $k$ given by
	\begin{equation}\label{eq:maxSINR}
		(\forall\vec{p}\in\stdset{R}_+^K)~u_k(\vec{p}) \eqdef \sup_{\substack{\rvec{v}_k\in\set{V}_k\\ \E[\|\rvec{v}_k\|_2^2] \neq 0}}\mathrm{SINR}_k(\rvec{v}_k,\vec{p}),
	\end{equation}
	where $\set{V}_k$ denotes a constraint modeling practical impairments on the feasible combiners, such as those related to imperfect CSI owing to measurement noise and pilot contamination, or to the enforcement of scalable cell-less regimes limiting the information sharing across the APs. Popular examples are the \textit{distributed} or \textit{centralized} \textit{user-centric} clustered models reviewed in \cite{demir2021}. Many relevant combining techniques are covered by \eqref{eq:maxSINR}, such as those involving the optimization of large-scale fading decoding coefficients given a fixed combining structure \cite{demir2021,bashar2019uplink}, or the optimization of the combining structure itself given the available CSI \cite{miretti2021team}. To avoid technical digressions, we refer the reader to \cite{miretti2021team,miretti22globecom} for details on these points. 
	
	\begin{figure}
		\centering
		\includegraphics[width=0.7\linewidth]{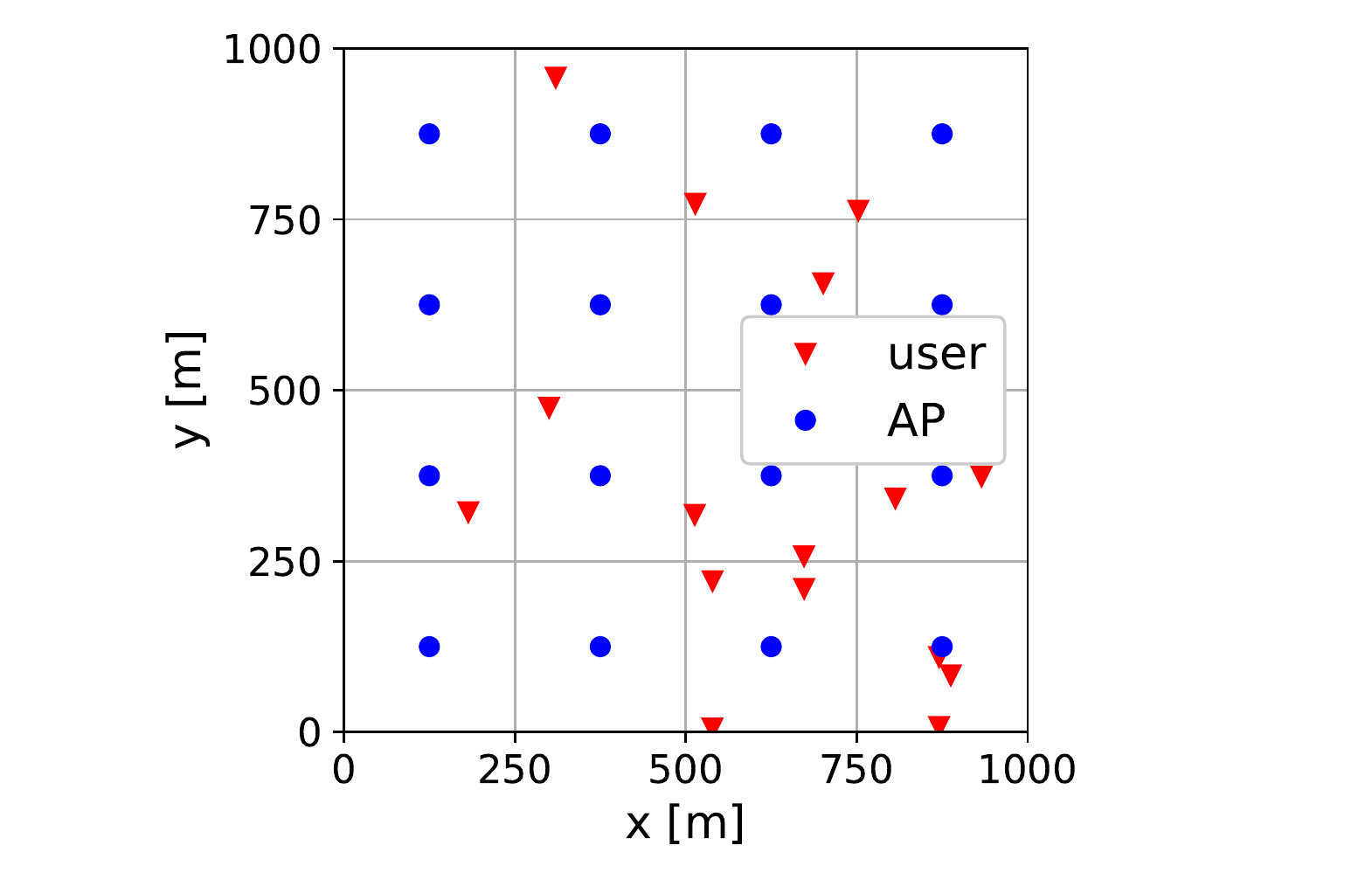}
		\caption{Pictorial representation of the simulated setup: $K=16$ users distributed within a squared service area of size $1\times 1~\text{km}^2$, and $L=16$ regularly spaced APs with $N=2$ antennas each. Each user is jointly served by a cluster of $4$ APs offering the strongest channel gains.}
		\label{fig:network}
	\end{figure}
	\subsection{Parametrizing the boundary of the achievable rate region}
	As recalled in Proposition~\ref{proposition.weak_par}, it is well-known that solving problem~\eqref{problem.maxmin} with the utilities in \eqref{eq:maxSINR} produces efficient SINR tuples; i.e., utilities on the weak Pareto boundary $\set{B}$. Following standard information theoretical terminology, a logarithmic transformation of the type $u \mapsto \log_2(1+u)$ of the set of efficient utilities readily gives the so-called boundary of the achievable rate region; i.e., the set of simultaneously achievable rate tuples under the given power and combining constraints. Network designers can tune the problem weights $(\omega_1,\ldots,\omega_K)$  to select rate tuples on the boundary of the achievable rate region according to their needs \cite{emil2014multiobjective,miretti22globecom}. However, at the price of  having limited prior control over the notion of fairness, Corollary~\ref{cor.main} shows that a simpler alternative for obtaining rate tuples that are efficient in the sense of Definition~\ref{def.def_paretto} is to choose any power allocation  $\vec{p}=(p_1,\ldots,p_K)$ such that $\|\vec{p}\|=\bar{p}$. To apply  this result formally, we only need to prove that the utilities in \eqref{eq:maxSINR} satisfy Assumption~\ref{assumption.restrictions}, which is done below by adding a minor condition that is always satisfied in nontrivial settings.  
	\begin{lemma} 
		Assume that $(\forall k~\in \set{K})~\set{V}'_k \eqdef \{\rvec{v}_k \in \set{V}_k~|~\E[\rvec{h}_k^\herm \rvec{v}_k]\neq 0\} \neq \emptyset$.
		Then, the utilities in \eqref{eq:maxSINR} satisfy Assumption~\ref{assumption.restrictions} with a positive concave function (hence a standard interference function) given by $(\forall k \in \set{K})(\forall\vec{p}\in\stdset{R}_+^K)$
		\begin{equation*} 
			\resizebox{\hsize}{!}{$f_k(\vec{p}) \eqdef \inf_{\rvec{v}_k \in \set{V}'_k} \dfrac{p_k\Var(\rvec{h}_k^\herm\rvec{v}_k)+\sum_{j\neq k}p_j\E[|\rvec{h}_j^\herm\rvec{v}_k|^2]+\E[\|\rvec{v}_k\|_2^2]}{|\E[\rvec{h}_k^\herm\rvec{v}_k]|^2}.$}
		\end{equation*} 
	\end{lemma}
	\begin{proof}
		See Proposition 3(ii) in \cite{miretti22globecom}.
	\end{proof}
	In the following numerical example we consider a per user power constraint; i.e., we focus on the $l_{\infty}$ norm constraint. This constraint is typical  for uplink operation. However, we remark that our result readily covers, for instance, a sum power constraint ($l_1$ norm constraint). If we invoke known uplink-downlink duality arguments \cite{demir2021}, this last special case can be used to study the downlink performance of the considered cell-less  network under a network-wide sum power constraint.\footnote{With an informal proof, this last special case has been recently reported  in \cite[Appendix~A]{miretti2021team} as a side motivation of the main results therein. The present studies provides a more rigorous and general proof.  
	}
	
	\subsection{Numerical example and final remarks}
	\begin{figure}
		\centering
		\includegraphics[width=\linewidth]{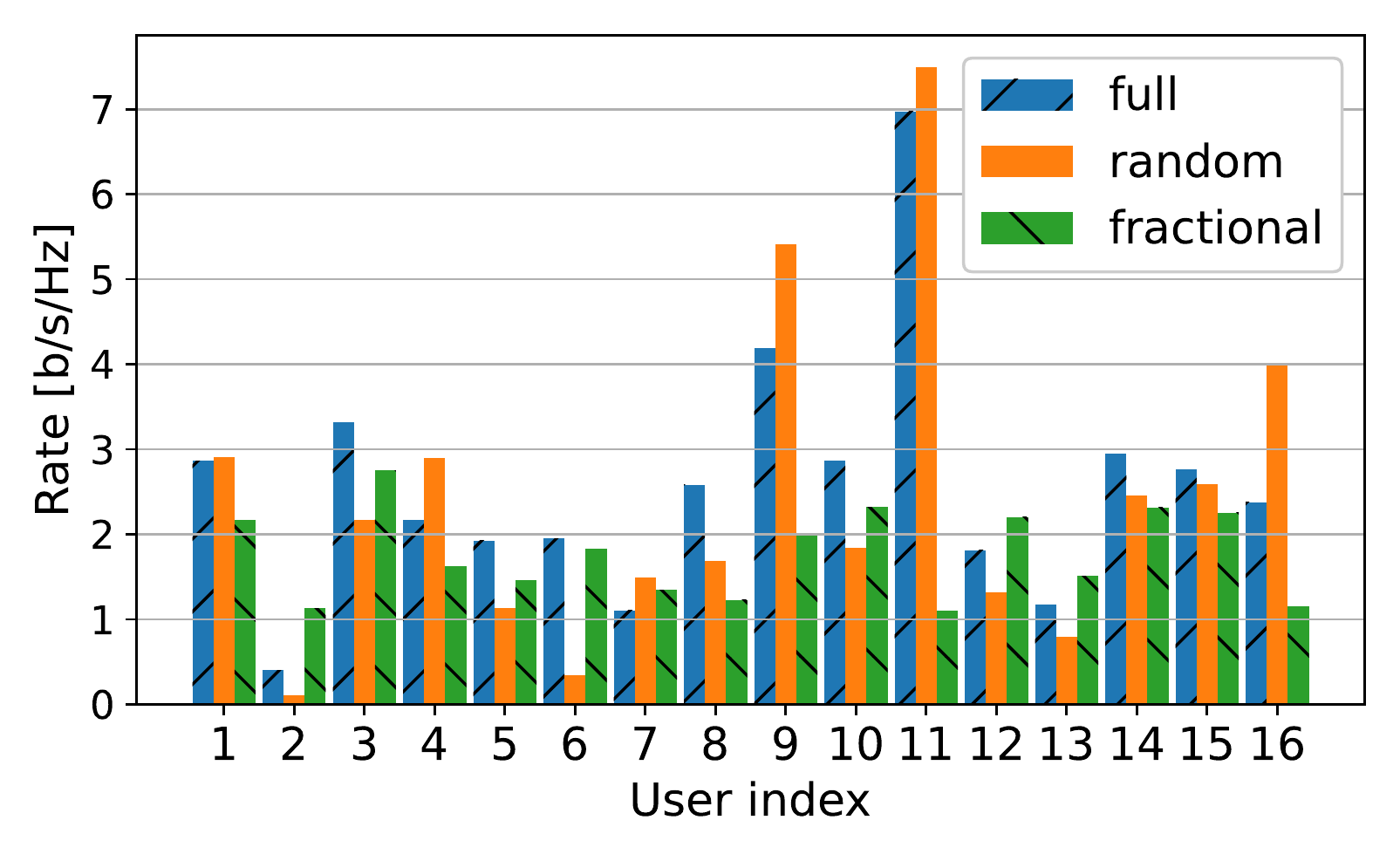}
		\caption{Ergodic achievable rates of all users under three different power allocation policies satisfying the per-user power constraint with equality. As predicted by Corollary~\ref{cor.main}, all chosen policies are Pareto efficient, i.e., they all produce rate tuples on the boundary of the achievable rate region.}
		\label{fig:hist}
	\end{figure}
	We simulate the cell-less network depicted in Figure~\ref{fig:network}. Except for the chosen values of $K$ and $N$, we use the same simulation parameters and channel model as in \cite[Sect.~IV]{miretti22globecom}. In particular, we consider a typical distributed user-centric clustered model as defined in \cite{demir2021}, where each user is jointly served by a cluster of four APs offering the strongest channel gain, and where each AP must form its combiners based on only \textit{local} CSI; i.e., without sharing any CSI along the fronthaul. The optimal combiners solving \eqref{eq:maxSINR} are obtained using the recently proposed \textit{team} theoretical technique discussed in \cite{miretti2021team}.  We omit further details owing to the space limitation.
	
	Figure~\ref{fig:hist} shows three achievable rate tuples obtained by choosing $\vec{p}$ according to the following power allocation policies: (a) full power transmission; (b) power transmission selected uniformly at random within a box $\{ \signal{p}\in\real^K~|~\|\signal{p}\|_\infty \le \bar{p}\}$; and (c) fractional power transmission with exponent $-1$ \cite[Eq.~(7.34)]{demir2021}. In all cases, the power allocation vector is later normalized such that a per user power constraint $\|\vec{p}\|_{\infty} \leq \bar{p} = 20$~dBm is satisfied with equality. 
	
	The numerical results are in agreement with Corollary~\ref{cor.main}, which states that all chosen policies are Pareto efficient. Indeed, we observe that no policy is able to offer strictly better rates for all users with respect to another policy. To the best of our knowledge, this result has not been reported  in the context of modern cell-less resource allocation problems involving the optimization of distributed combiners. The difference between the three policies lies in the notion of fairness. For instance, we observe that policy (a) offers very high rates to users with good channel conditions, but at the expense of weak users. In contrast, policy (c) approximates an egalitarian max-min fair policy that maximizes the rate of the weakest user.

	Finally, we emphasize that the results in Sect.~\ref{sect.wpar} cover not only the cell-less resource allocation problems described in this section, but also  many others, including those in \cite[Sect.~4.1.4]{brehmer2012utility}, \cite[Sect. 3.2.3]{bjornson2013resource} as particular cases.
	
	
	\bibliographystyle{IEEEtran}
	\bibliography{IEEEabrv,references}

\end{document}